\documentclass[12pt, draftclsnofoot, onecolumn]{IEEEtran}
\usepackage{amsfonts, amsthm}
\usepackage{amsmath, color}
\usepackage{graphicx}
\usepackage{cite}
\usepackage[normalem]{ulem}
\usepackage{tikz}
\usepackage{epstopdf}
\usepackage{psfrag}
\usepackage{etoolbox}
\usepackage[process=auto]{pstool}
\usetikzlibrary{arrows,automata}
\usepackage{graphicx}
\usepackage{amssymb}
\usepackage{hyperref}
\usepackage{siunitx}

\newtheorem{theorem}{Theorem}

\newtheorem{corol}{Corollary}

\newtheorem{remark}{Remark}

\allowdisplaybreaks
\title{Diffusive Molecular Communication in a Biological Spherical Environment with Partially Absorbing Boundary}
\author{Hamidreza~Arjmandi, Mohammad~Zoofaghari, and Adam~Noel}
\begin{document}
\maketitle
%\author{\IEEEauthorblockN{1\textsuperscript{st} Ali Etemadi}
%\IEEEauthorblockA{\textit{Dep. of Electrical and Computer Eng.} \\
%\textit{Tarbiat Modares University}\\
%Tehran, Iran \\
%ali.etemadi@modares.ac.ir}
%\and
%\allowdisplaybreaks
%
%\title{Diffusive Molecular Communication in Biological Cylindrical Environment}
%\author{Mohammad Zoofaghari, Hamidreza Arjmandi}
%%$^*$ Sharif~University~of~Technology, $^{\dagger}$~Friedrich-Alexander~University~of~Erlangen-Nuremberg
%\thanks{M. Zoofaghari and H. Arjmandi are with the Electrical Engineering Department, Yazd University, Yazd, Iran (e-mails: \{zoofaghari, arjmandi\}@yazd.ac.ir).}
%\maketitle

%\author{\IEEEauthorblockN{Mohammad Zoofaghari, Hamidreza Arjmandi}
%\IEEEauthorblockA{{Department of Electrical Engineering, Yazd University} \\
%zoofaghari@yazd.ac.ir, arjmandi@yazd.ac.ir}}
%\and
%\IEEEauthorblockN{3\textsuperscript{rd} Paeiz Azmi}
%\IEEEauthorblockA{\textit{Dep. of Electrical and Computer Eng.} \\
%\textit{Tarbiat Modares University}\\
%Tehran, Iran \\
%pazmi@modares.ac.ir}
%\and
%\IEEEauthorblockN{4\textsuperscript{th} Nader Mokari}
%\IEEEauthorblockA{\textit{Dep. of Electrical and Computer Eng.} \\
%\textit{Tarbiat Modares University}\\
%Tehran, Iran \\
%nader.mokari@modares.ac.ir}
%}

\begin{abstract}
Diffusive molecular communication (DMC) is envisioned as a promising approach to help realize healthcare applications within bounded biological environments. In this paper, a DMC system within a biological spherical environment (BSE) is considered, inspired by bounded biological sphere-like structures throughout the body. As a biological environment, it is assumed that the inner surface of the sphere's boundary is fully covered by biological receptors that may irreversibly react with hitting molecules. Moreover, information molecules diffusing in the sphere may undergo a degradation reaction and be transformed to another molecule type. Concentration Green's function (CGF) of diffusion inside this environment is analytically obtained in terms of a convergent infinite series. By employing the obtained CGF, the information channel between transmitter and transparent receiver of DMC in this environment is characterized. Interestingly, it is revealed that the information channel is reciprocal, i.e., interchanging the position of receiver and transmitter does not change the information channel. Results indicate that the conventional simplifying assumption that the environment is unbounded may lead to an inaccurate characterization in such biological environments.

\end{abstract}


\section{Introduction}

Diffusive molecular communication (DMC) is a promising approach for realizing nano-scale communications \cite{Nakano13}. In DMC, molecules are used to carry information from transmitter to receiver nanomachine via a diffusion mechanism. Information is encoded in the concentration, type, and/or release time of molecules. %In particular, a transmitter nanomachine releases information molecules into the environment. The released molecules move randomly via Brownian motion and may be observed at the receiver \cite{Pierobon2011}.
Due to the potential of bio-compatibility \cite{Akyldiz2011}-\cite{Farsad16}, DMC is envisioned to be widely applied in healthcare applications.
In biological environments, the DMC system may be influenced by various environmental properties, e.g., bounded environment and/or degradation reactions. The effects of these characteristics need to be accounted for in the analysis of DMC system performance.

%A biological organ which releases molecules of the same type of information molecules plays the role of an external noise source. In \cite{}, the authors have modeled the random time and amplitude of molecule release from a biological organ as compound Poisson process. Correspondingly, the received noise has been derived as a compound Poisson noise source.
%Also, the biological organs may impose different boundary conditions for DMC systems, e.g. absorbing, reflective, or partially absorbing boundary conditions.
The performance of DMC systems with different geometries and environmental boundary conditions have been investigated in the literature. DMC system performance in an ideal unbounded environment has been extensively studied in communication engineering \cite{Kuran10,Pier14,Bazar14,Noel14,Aijaz15,Kilinc13,Wang15,Mahfuz14,Pier11}. This assumption leads to simple tractable analysis of diffusion which may provide insightful ideas about the effect of different parameters. %Diffusive molecules in an unbounded environment are dispersed in all directions. Thereby, the range of the communication is very limited in unbounded environment. Obtaining a longer range of communication requires high release rate of molecules which may be impossible for nanomachines with constrained resource.
However, the unbounded environment is generally not a realistic assumption for an in-vivo environment. Thereby, different bounded environment models have also been proposed for DMC systems.
%Also, diffusion in a spherical environment has also been investigated in some works.

Inspired by the blood vessel structures and microfluidic channels, bounded cylindrical environments have been considered for DMC systems.
In \cite{Farsad12}, diffusion communication channels inside a cylindrical environment with elastic (i.e., reflective) walls was characterized via particle-based simulation. In \cite{Kuran13}, a cylindrical DMC model with absorbing walls and no flow was considered. The hitting times and probabilities were obtained from simulation results.
The response to a pulse of carriers, released by a mobile transmitter, was measured by receivers positioned over the vessel wall in \cite{Feli13}, also based on simulation results. In \cite{Turan18}, a cylindrical DMC environment was considered where the receiver partially covers the cross-section of a reflective cylinder. The distribution of hitting locations is again obtained from simulation results.
%Papers \cite{Farsad12}-\cite{Turan18} characterize MC channel only based on the simulation results without providing mathematical analysis.
The authors in \cite{Wayan17}  considered the diffusion in a cylinder with reflective walls and non-uniform fluid flow. Assuming a transmitter point source, the channel impulse responses for two simplifying flow regimes referred to as dispersion and flow-dominant, were derived.
In \cite{Dinc18}, the authors obtained the channel impulse response for a 3-D microfluidic channel environment in the presence of flow where the boundaries are reflective. Also, in \cite{Zoofaghari18} we obtain the concentration Green's function in a biological cylindrical environment where the boundary is covered by receptor proteins and information molecules are subject to both flow and chemical degradation.

Another useful and relevant geometry is the bounded spherical environment, as considered in  \cite{Alzubi18,Dinc182} which is inspired by some sphere-like entities in the body, e.g., stomach, lung, kidney, cells, nucleus, etc.
In  \cite{Alzubi18,Dinc182}, the outer boundaries are idealized as fully absorbing and fully reflective boundaries, respectively, and the receiver is assumed to be located at the center of the bounded sphere.
In \cite{Alzubi18}, the authors consider a DMC system in a bounded sphere whose boundary is fully absorbing and a spherical receiver is assumed to be located at the center of the sphere and its surface is covered by ligand receptors.
%By assuming a point source transmitter at an arbitrary location, the channel impulse response is analyzed.
%In \cite{Dinc182}, a spherical DMC environment with a reflective boundary, a spherical absorbing receiver, and a point source transmitter was considered. For this system model, the channel impulse response was derived. To obtain the channel impulse response in both  \cite{Alzubi18}-\cite{Dinc182}, an ideal symmetric topology was adopted where the receiver is located at the center of the bounded sphere.
Therefore, the analysis proposed by these works cannot account for the diffusion asymmetry in the elevation and azimuth coordinates, which may be unavoidable depending on the locations of the transmitter and receiver.
Furthermore, the simplifying boundary conditions
%(i.e., an absorbing environment boundary in \cite{Alzubi18} and a reflective boundary in \cite{Dinc182})
may not hold in-vivo environments, where boundaries covered by biological receptors may lead to partial absorption of molecules. For instance, the inner surface of many internal organs such as the stomach and the lung are coated with epithelial cells. Also, the inner layer of blood vessels is surrounded by endothelium cells \cite{Cliff}. The surfaces of these cells contain various types of receptors and act as an interface between the underlying layer and the outside environment.

%As a more compatible model for the blood vessel structures in the body and microfluidic channels, a bounded cylindrical environment can be employed. Investigation of MC in the cylindrical environment adopting various assumptions has been considered in the literature.

%The authors in \cite{Alzubi18} consider a DMC system in bounded spherical environment  where the transmitter is a point source, the receiver is located at the center of the sphere and covered by ligan receptors, and boundary is assumed absorber. For this model, the expected number of activated receptors (or number of ligandreceptor complexes) has been analyzed.

%As another important mechanism in in-vivo environment, the effect of degradation reactions has not been analytically investigated for DMC in the cylindrical environment, in the previous works. Interestingly, the degradation reactios may be utiliazed to overcome intersymbol interference (ISI) of the diffusion channel \cite{Feli13},\cite{Adam14} which is inspired from living organisms. For instance, Acetylcholinesterase molecules destroy the messenger Acetylcholine molecules in the channel between nerve cell and muscle cell in neuromuscular junction to clean the channel for the next signal transmission \cite{Nelson}.

In this paper, we consider a point-to point DMC system in a bounded biological spherical environment (BSE). The inner layer of the outer environment boundary is assumed to be covered with biological receptors, leading to a partially absorbing boundary. An information molecule can act as a ligand if it hits the boundary and reacts with a receptor molecule to produce a ligand-receptor complex. A simple irreversible ligand-receptor reaction is considered to make our analysis analytically tractable. Moreover, a degradation reaction is assumed within the environment such that the diffusive information molecules may be transformed into another type. %Furher, a uniform flow \cite{Farsad12} in the cylinder is assumed. Although, flow with constant velocity is not a realistic assumption in cylindrical environment, it makes possible to obtain analytic solution for corresponding diffusion equations which provides insightful ideas about the effect of different system parameters.

Assuming a point source transmitter at an  arbitrary location in the sphere, we analytically obtain the Green's function (CGF) of diffusion inside this environment as a convergent infinite series that accounts for the asymmetry in all radial, elevation, and azimuthal directions.
A point-to-point DMC system is considered within a BSE where the point source transmitter and a transparent receiver are at arbitrary locations. %Thereby, they do not affect the Brownian motion of molecules. An arbitrary transmitter geometry (not necessarily point source) with arbitrary transmitted modulated signal (not necessarily impulsive release signal) is adopted.
By employing the obtained CGF, the probability density function (PDF) for observation times of a molecule at the receiver is characterized. Correspondingly, the average received signal at the observing receiver is derived. Interestingly, the obtained expression for the CGF reveals the channel reciprocity, i.e., interchanging the positions of the receiver and transmitter does not change the CGF and correspondingly the average received signal. Furthermore, the stochasticity of the received signal is analyzed and accordingly the information channel between the transmitter and receiver is characterized.
The proposed analysis is confirmed by particle-based simulation (PBS) results. Also, the effect of system parameters on the observation time PDF are examined. Our results indicate that the conventional ideal assumption of an unbounded environment  may lead to an inaccurate characterization of BSE. %The presented model and analysis can be used to predict the drug concentration profile in biological sphere-like organs for drug delivery applications.

%To evaluate the proposed DMC system model, a simple on-off keying modulation scheme is considered and corresponding error probability is obtained. %Obviously, the presented model and analysis can be used to predict the concentration profile of drug in the blood vessels for healthcare applications.

%It is observed that the degradation reaction and partially absorbing boundary may be utilized to mitigate ISI and outperform corresponding error probability.
%Moreover, our results reveal that the analysis proposed for the CGF with constant velocity flow well approximates CGF obtained from PBS with more realistic Poiseuille flow model, for enough small velocity values.

The paper is organized as follows. The system model is presented in Section II. The CGF of diffusion in the BSE is obtained in Section III. In Section IV, the information channel between the transmitter and receiver is characterized, and the error probability of DMC with a simple on-off keying modulation over this channel is presented. The simulation and numerical results are presented in Section V. Finally, the paper is concluded in Section VI.

 % computing the optimal threshold is computationally complex since nanomachines have limited resources which constrains the affordable complexity, a suboptimal simple STD is proposed and an efficient recursive algorithm for finding the decision threshold is developed.
%The corresponding BER is an upper bound for error probability of the DMC system.

%For the special case of a CPNS in the high rate\footnote{The ``rate of the CPNS'' characterizes the average number of the release events in a time unit.} regime, the non-uniform sampling theorem is applied to obtain the statistics of the received number of molecules. Interestingly, from the receiver perspective, a CPNS in the high rate regime approaches a homogeneous Poisson noise source. %It is shown that for the special case of CPNS in high rate regime the optimal detector is a single-threshold detector where the optimal threshold is obtained analytically.

%The remainder of this paper is organized as follows: In Section \ref{section2}, we present the DMC system model including the  transmitter, receiver, channel, and  CPNS models. In Section \ref{section3}, the distributions of the received signals due to the release of molecules by the transmitter and the CPNS are derived. The optimal ML detector and the error probability of the DMC system in the presence of a CPNS are analyzed in  Section \ref{section4}. In Section \ref{section5}, we provide simulation and numerical results. Finally, the paper is concluded in Section \ref{section6}.

\section{System Model} \label{system model} \label{section2}
\subsection{Biological Spherical Environment}
 The spherical coordinate system is used to describe the environment geometry where $(r,\theta,\varphi)$ denote radial, elevation, and azimuth coordinates, respectively. A sphere with radius $r_s$ is considered and the center of the sphere is chosen as the origin of the coordinate system.
%\begin{equation}
%	r=r_s,\;\;\; 0<\theta<\pi, \;\;\; 0\leq \varphi<2\pi
%\end{equation}
The following degradation reaction is considered in the environment in which the (information) molecules $A$ diffusing in the environment may be transformed to another molecule type:
 \begin{equation}\label{deg1}
	\mathrm{A} \overset{k_{\mathrm d}}{\to} \mathrm{\hat{A}},
\end{equation}
where $k_{\mathrm d}$ is the degradation reaction constant in $\mathrm{s}^{-1}$. % and molecule $\mathrm{\hat{A}}$ is not recognized by the receiver.
We assume that the sphere boundary is fully covered by infinitely many biological receptors. An A molecule (ligand) hitting the boundary may bind to a receptor (R) and produce a ligand-receptor complex (AR). A simple irreversible reaction for the receptors on the boundary is considered as follows:
%The boundary is assumed to be covered by receptor proteins where hitting molecules (ligands) may react with the receptors on the boundary. Depending on the forward reaction rate constant $k_{\mathrm f}$, the molecule  hitting to the receptor may be reflected or activates the receptor and produce the ligand-receptor complex ($AR$). \textbf{We assume the backward reaction rate constant is zero and a molecule that bind to the receptor will not return to the environment. Therefore, we have
 \begin{equation} \label{deg2}
	\mathrm{A}+\mathrm{R} \overset{k_{\mathrm f}}{\to} \mathrm{AR},
\end{equation}
where $k_{\mathrm f}$ is the forward reaction constant in $\si{m.s^{-1}}$. Thus, the boundary is partially absorbing where a hitting molecule is absorbed with a probability dependent on $k_{\mathrm f}$. The boundary has the special cases of purely reflective and perfectly absorbing for $k_{\mathrm f}=0$ and $k_{\mathrm f}=\infty$, respectively. In this paper, the effect of receptor occupancy is neglected and the formations of the individual ligand-receptor complexes are assumed to be independent of each other. As a result, multiple information molecules may react within the same vicinity of each other on the sphere boundary and at the same time.

%A flow with constant velocity $v$ $ms^{-1}$ \cite{Farsad12} in axial direction is considered inside the cylinder, i.e., the velocity field is given by $\bar{v}(\bar r)=v \hat a_z $ $ms^{-1}$. %Although, flow with constant velocity is not a realistic assumption, it makes possible to obtain analytic solution for corresponding diffusion equations which provides insightful ideas about the effect of different parameters.
%Moreover, as our results indicate the proposed analysis for the CGF with constant velocity flow in the next section well approximates CGF obtained from PBS with nonuniform Poiseuille flow model, for enough small velocity values.
\subsection{DMC System in BSE}
A point-to-point DMC system is considered within the bounded biological spherical environment. A point source transmitter located at an arbitrary point $\bar{r}_{\rm tx}=(r_{\rm tx},\theta_{\rm tx},\varphi_{\rm tx})$ in the sphere is assumed. The transmitter uses information molecules of type A. The diffusion coefficient of the medium for information molecule $A$ is denoted by $D$ \si{m^2.s^{-1}}.
Also, a transparent receiver is considered that does not affect the Brownian motion of molecules. The receiver is a sphere with radius $R_{\rm rx}$ whose center is located at $\bar{r}_{\rm rx}=(r_{\rm rx},\theta_{\rm rx},\varphi_{\rm rx})$.
%\begin{equation}
%	\rho_1<\rho<\rho_2,\; \varphi_1<\varphi<\varphi_2 ,\;  z_1<z<z_2 .
%\end{equation}
  %Obviously, the volume of the receiver equals to $V_{\rm rx}=\varphi_0 \rho_0^2 z_0$.
%
%with radius $r_{\rm rx}$ is assumed whose center is at arbitrary location of $(\rho_{\rm rx},\varphi_{\rm rx},z_{\rm rx})$. Therefore, the distance between receiver and transmitter equals to
%\begin{equation}
%d=\\
%\sqrt{(\rho_{\rm rx} \cos{ \varphi_r}-\rho_{\rm tx} \cos{\varphi_{\rm tx}})^2+(\rho_{\rm rx}  \sin (\varphi_{\rm rx})-\rho_{\rm tx} \sin{\varphi_{\rm tx}})^2+(z_{\rm rx}-z_{\rm tx})^2}
%\end{equation}
A schematic illustration of the system model is represented in Fig. \ref{Fig0}
\begin{figure}
\center
\includegraphics[width=15 cm,height=9 cm]{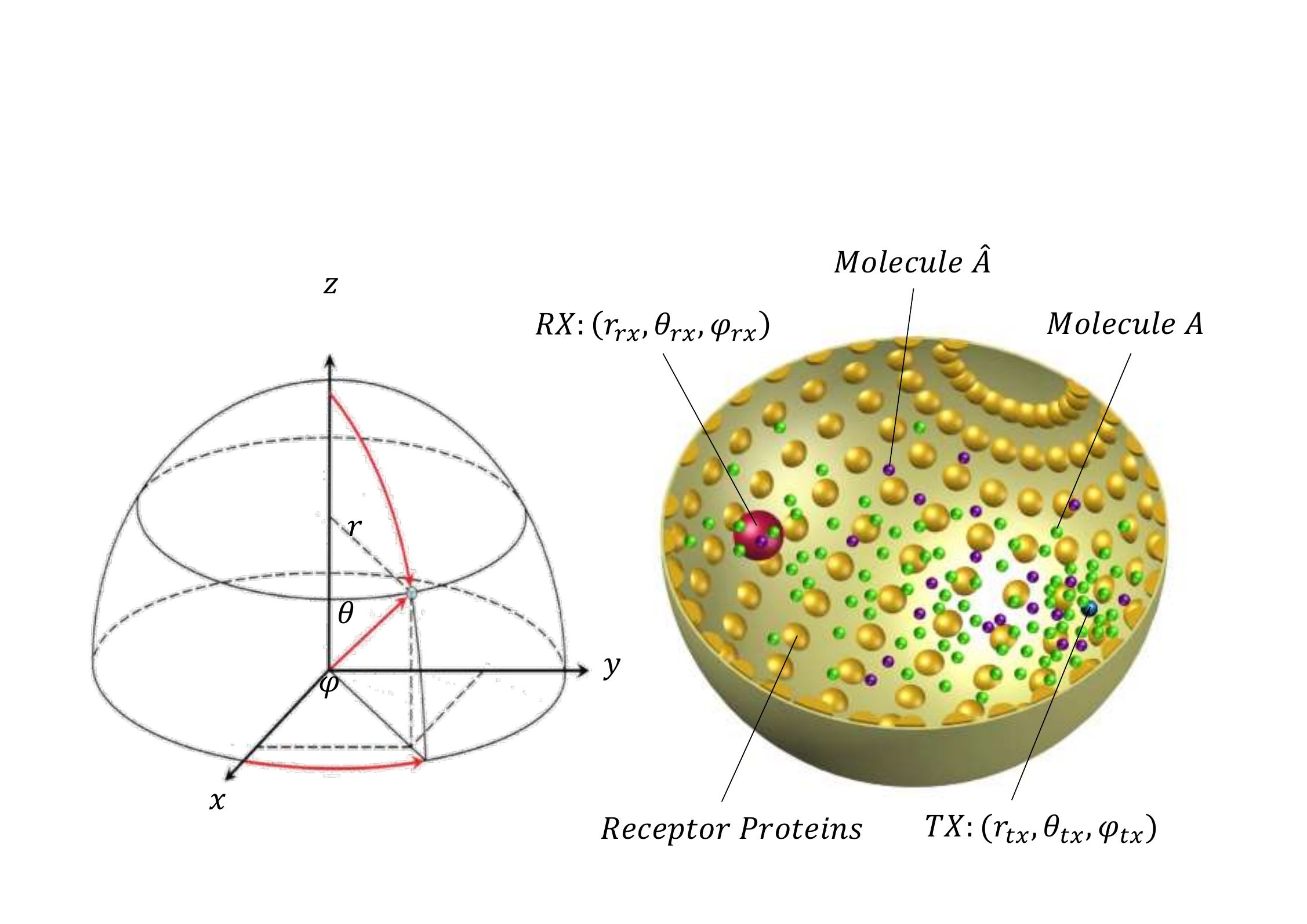}
	\setlength{\abovecaptionskip}{-0.5 cm}
 \caption{DMC system in biological spherical environment (Only one hemisphere has been illustrated).}
\label{Fig0}
\end{figure}

Time is divided into time slot durations of $T$ seconds (s). The receiver and transmitter are assumed to be perfectly synchronized \cite{syncArj}. In each time slot, the transmitter releases information molecules into the environment according to the intended symbol. The released molecules move randomly in the environment following Brownian motion. Their movements are assumed to be independent of each other. The diffusing molecules, which are exposed to the degradation reaction and binding with the receptors on the boundary, may be observed at the receiver at a sampling time. The receiver counts the number of molecules within its volume at the sampling time to decide the intended transmitted symbol.
To analyze the presented DMC system, we formulate the Green's function boundary value problem for diffusion in the described environment.
\subsection{Green's function Boundary Value Problem}
We assume that the point source transmitter, located at an arbitrary point $\bar{r}_{\rm tx}=(r_{\rm tx},\theta_{\rm tx},\varphi_{\rm tx})$ inside the sphere has an instantaneous molecule release rate of $\delta(t-t_0)$ molecule $(\rm mol)/ \rm s$, where $\delta(\cdot)$ is Dirac delta function. In the spherical coordinate system, this impulsive point source can be represented by the function $S(\bar r,t,{\bar{r}_{\rm tx}},t_0)=\frac{{\delta (r  - {r _{\rm tx}})\delta(\theta-\theta_{\rm tx})\delta(\varphi-\varphi_{\rm tx})\delta(t-t_0)}}{ r^2\sin\theta  }$ $\si{mol.s^{-1}.m^{-3}}$. Given the source $S(\bar r,t,{\bar{r}_{\rm tx}},t_0)$ and the degradation reaction \eqref{deg1}, the molecular diffusion is described by partial differential equation (PDE) \cite{Grindrod}
\begin{align}\label{fick}
D{\nabla ^2}C(\bar r,t|{{\bar r}_{\rm tx}},{t_0})
 - {k_{\mathrm d}}C(\bar r,t|{{\bar r}_{\rm tx}},{t_0}) + S(\bar r,t,{\bar{r}_{\rm tx}},t_0) = \frac{{\partial C(\bar r,t|{{\bar r}_{\rm tx}},{t_0})}}{{\partial t}}
\end{align}
where $C(\bar r,t|{{\bar r}_{\rm tx}},{t_0})$ denotes the molecule concentration at point $\bar r$ and time $t$.
In the spherical coordinate system, \eqref{fick} is re-written as
\begin{align}\label{Eq}
  \frac{D}{{{r^2}}}\frac{\partial }{{\partial r}}\left({r^2}\frac{{\partial C(\bar r,t|{{\bar r}_{\rm tx}},{t_0})}}{{\partial r}}\right) + \frac{D}{{{r^2}\sin \theta }}\frac{\partial }{{\partial \theta }}\left(\sin \theta \frac{{\partial C(\bar r,t|{{\bar r}_{\rm tx}},{t_0})}}{{\partial \theta }}\right) + \frac{D}{{{r^2}{{\sin }^2}\theta }}\frac{{{\partial ^2}C(\bar r,t|{{\bar r}_{\rm tx}},{t_0})}}{{\partial {\varphi ^2}}}\\
 - {k_{\mathrm d}}C(\bar r,t|{{\bar r}_{\rm tx}},{t_0})+\frac{{\delta (r  - {r _{\rm tx}})\delta(\theta-\theta_{\rm tx})\delta(\varphi-\varphi_{\rm tx})\delta(t-t_0)}}{ r^2\sin\theta  } = \frac{{\partial C(\bar r,t|{{\bar r}_{\rm tx}},{t_0})}}{{\partial t}}. \nonumber
\end{align}
The irreversible ligand-receptor reaction over the sphere boundary given in \eqref{deg2} is characterized by the third type (Robin) boundary condition of \cite{Crank} \footnote{Since the condition is over the inner boundary, i.e.,  $C(\bar r,t|{{\bar r}_{\rm tx}},{t_0})$ is the concentration for $r\leq r_s$, the negative sign on the right side is required.}
\begin{equation}\label{BD1}
D\frac{{\partial C(\bar r,t|{{\bar r}_{\rm tx}},{t_0})}}{{\partial r }}\mid_{\bar r=(r_s,\theta,\varphi)}=-k_{\mathrm f} C(r_s,\theta,\varphi,t|{{\bar r}_{\rm tx}},{t_0}).
\end{equation}
The concentration function $C(\bar r,t|{{\bar r}_{\rm tx}},{t_0})$ that satisfies \eqref{fick} subject to the boundary condition \eqref{BD1} is called the concentration Green's function (CGF) of diffusion. Using Green's function, the solution of diffusion for an arbitrary source can be obtained based on the superposition principle.

%The information channel between the transmitter and receiver in this environment is characterized, correspondingly. Then, the stochastic of the received signal at the receiver is investigated. In order to evaluate the performance of the system in terms of error probability, a simple on-off keying modulation is adopted. Bits $1$ and $0$ are represented by the instantaneous release of $N$ molecules (on average) and no molecule, respectively.

%In order to evaluate the DMC system performance based on proposed analysis, a simple on-off keying modulation is adopted where bits 1 and 0 are represented by the release of $N$ molecules (on average) and no molecule at the beginning of  each time slot, respectively. %Assuming  transmission of bit 1, the number of molecules released by the transmitter follows  a Poisson distribution with mean $N$ \cite{Arjmandi2013}. The receiver is assumed to be a transparent spherical volume of radius $r_R$ that counts the number of molecules inside the receiver volume  at sampling time $t_s$ \cite{schober2015}. The receiver uses the observed sample to decide about the  transmitted bit. In the rest of this section, the  CPNS model is presented and the adopted channel model is described.

\section{Deriving CGF and Characterizing Received Signal}
In this section, we derive the Green's function for diffusion in a bounded BSE.
\subsection{CGF of Diffusion in Biological Spherical Environment}

%We have also the implicit axial unbounded conditions that concentration is zero when $z\to \pm\infty$, i.e.,
%\begin{equation}\label{BD2}
%C(\bar r,t|{{\bar r}_{\rm tx}},{t_0})\mid_{\bar r=(\rho_c,z=\pm \infty,\varphi)}=0.
%\end{equation}

The impulsive point source in \eqref{Eq} is equivalent to considering an initial condition of
\begin{equation} \label{sepin}
C(\bar r,t=t_0|{{\bar r}_{\rm tx}},{t_0})=\frac{{\delta (r  - {r _{\rm tx}})\delta (\theta-\theta_{\rm tx})\delta (\varphi  - {\varphi _{\rm tx}})}}{r^2 \sin\theta}.
\end{equation}
By considering this initial condition and removing the source term in \eqref{Eq}, we obtain a homogeneous PDE that can be solved by the well-known technique of separation of variables {\cite{HG}}. The solution with separated variables is considered as follows:
\begin{equation}\label{Eqs}
  C(r,t|{r_{\rm tx}},{t_0}) = R(r|r_{\rm tx})\Theta (\theta|\theta_{\rm tx} )\Phi (\varphi|\varphi_{\rm tx} )T(t|t_0).
\end{equation}
By substituting \eqref{Eqs} into PDE \eqref{Eq} and the boundary condition \eqref{BD1}, dividing both sides of equalities by $R(r|r_{\rm tx})\Theta (\theta|\theta_{\rm tx} )\Phi (\varphi|\varphi_{\rm tx} )T(t|t_0)$, and with some simple manipulation, we have
\begin{align}\label{sep}
 &{r^2}{\sin ^2}\theta \left(\frac{2}{r}\frac{{R'(r|{r_{\rm tx}})}}{{R(r|{r_{\rm tx}})}} + \frac{{R''(r|{r_{\rm tx}})}}{{R(r|{r_{\rm tx}})}}\right) + \sin \theta \cos \theta \frac{{\Theta '(\theta |{\theta _{\rm tx}})}}{{\Theta (\theta |{\theta _{\rm tx}})}} + \\
&{\sin ^2}\theta \frac{{\Theta ''(\theta |{\theta _{\rm tx}})}}{{\Theta (\theta |{\theta _{\rm tx}})}} - \frac{{T'(t|{t_0})}}{{DT(t|{t_0})}}{r^2}{\sin ^2}\theta  - \frac{{{k_{\mathrm d}}}}{D}{r^2}{\sin ^2}\theta  = - \frac{{\Phi ''(\varphi |{\varphi _{\rm tx}})}}{{\Phi (\varphi |{\varphi _{\rm tx}})}}\overset{(a)}{=}\alpha\nonumber
\end{align}
subject to the following boundary condition:
\begin{equation}\label{sep3B}
   DR'(r|r_{\rm tx})\mid_{r=r_s} = -{k_{\mathrm f}}R(r_s|r_{\rm tx}),
 \end{equation}
where equality with constant $\alpha$ in (a) holds, since we have two separated functions on the left and right hand sides of the first equality.
From \eqref{sep}, we have the following ordinary differential equation:
\begin{equation}\label{sep1}
\Phi''(\varphi|\varphi_{\rm tx})+\alpha \Phi(\varphi|\varphi_{\rm tx})=0.
\end{equation}
The concentration function is a symmetric function with respect to $\varphi=\varphi_{\rm tx}$. Thus, the possible solution of \eqref{sep1} is
\begin{equation}\label{Phi}
  \Phi_m(\varphi|\varphi_{\rm tx})=G_m \cos(\sqrt{\alpha}(\varphi-\varphi_{\rm tx})),
\end{equation}
where $G_m$ is an unknown constant.
The concentration function is also periodic with period $2\pi$ with respect to the $\varphi$ variable. Thus,  $\alpha=m^2$ is acceptable for all non-negative integer values of $m\in \mathbb{Z}_{+}$.

Considering the equality of the left hand side of \eqref{sep} with $\alpha=m^2$ and some simple manipulation, we obtain
\begin{align}\label{sep2}
&{r^2}(\frac{2}{r}\frac{{R'(r|{r_{\rm tx}})}}{{R(r|{r_{\rm tx}})}} + \frac{{R''(r|{r_{\rm tx}})}}{{R(r|{r_{\rm tx}})}}) - \frac{{T'(t|{t_0})}}{{DT(t|{t_0})}}{r^2} - \frac{{{k_{\mathrm d}}}}{D}{r^2} = \\
& - \frac{{\cos \theta }}{{\sin \theta }}\frac{{\Theta '(\theta |{\theta _{\rm tx}})}}{{\Theta (\theta |{\theta _{\rm tx}})}} - \frac{{\Theta ''(\theta |{\theta _{\rm tx}})}}{{\Theta (\theta |{\theta _{\rm tx}})}} + \frac{{{m^2}}}{{{{\sin }^2}\theta }}\overset{(b)}{=}\beta, \nonumber
\end{align}
where equality with constant $\beta$ in (b) holds, since we have two separated functions on the left and right hand sides of the first equality.
By defining $\beta=\nu(\nu+1)$ where $\nu$ is a real number, and with simple manipulations of the second equation in \eqref{sep2}, we obtain
 \begin{equation}\label{sep3}
  \Theta ''(\theta |{\theta _{\rm tx}}) + \frac{{\cos \theta }}{{\sin \theta }}\Theta '(\theta |{\theta _{\rm tx}}) + \left(\nu(\nu + 1) - \frac{{{m^2}}}{{{{\sin }^2}\theta }}\right)\Theta (\theta |{\theta _{\rm tx}}) = 0,
 \end{equation}
which is the well-known Legendre equation \cite{Mandelis}. The principal solution for \eqref{sep3} is given by
\begin{equation}\label{sep3s}
  \Theta (\theta ) = AP_\nu^m(\cos \theta ) + BQ_\nu^m(\cos \theta )
\end{equation}
where $P_\nu^m(\cdot)$ and $Q_\nu^m(\cdot)$ are the associated Legendre functions of the first and second kind, respectively with degree $\nu$ and order $m$ \cite{Mandelis}. Since $Q_\nu^m(\cos\theta)$ is singular at $\theta=0$ and $\theta=\pi$ for all values of $\nu$, we set $B=0$. Also, for non-integer values of $\nu$, $P_\nu^m(cos\theta)$ is singular at $\theta=\pi$. Therefore, $\Theta_{nm}(\theta|\theta_{\rm tx})=A_{nm}P_n^m(\cos \theta )$ is an acceptable solution of \eqref{sep3s} for each integer value of $n \in \mathbb{Z}$, where $A_{nm}$ is an unknown constant. Because of the linear dependency of $P_n^m(\cos \theta )=P_{-n-1}^m(\cos \theta )$, only \begin{equation}\label{Theta}
 \Theta_{nm}(\theta|\theta_{\rm tx})=A_{nm} P_n^m(\cos \theta )
  \end{equation}
   for non-negative integer values $n\in \mathbb{Z}_{+}$ are linearly independent solutions for \eqref{sep3s}.

Considering $\beta=n(n+1)$ in \eqref{sep2} and performing some simple manipulations, we obtain
\begin{equation}\label{sep4}
  D(\frac{2}{r}\frac{{R'(r|{r_{\rm tx}})}}{{R(r|{r_{\rm tx}})}} + \frac{{R''(r|{r_{\rm tx}})}}{{R(r|{r_{\rm tx}})}}) - D\frac{{n(n + 1)}}{{{r^2}}} = \frac{{T'(t|{t_0})}}{{T(t|{t_0})}} + k_{\mathrm d} \overset{(c)}{=}\gamma
\end{equation}
where equality with constant $\gamma$ in (c) holds, since we have two functions with separated variables on the left and right hand sides of the first equality.
%Note that only a negative constant $\gamma$ on the right side is possible, since a positive constant leads to unbounded function $T(t|t_0)$ and correspondingly unbounded concentration function of time.
%We set $\gamma=-{D}\lambda_n^2$
From \eqref{sep4}, we have
\begin{equation}\label{sep5}
{r^2}R''(r|{r_{\rm tx}}) + 2rR'(r|{r_{\rm tx}}) + ({-\frac{\gamma}{D}}{r^2} - n(n + 1))R(r|{r_{\rm tx}}) = 0,
\end{equation}
which is the Bessel equation \cite{Beals}. The solution for \eqref{sep5} should satisfy the boundary condition given in \eqref{sep3B}, because it also includes $R(r|r_tx)$.
%Positive $\gamma$ values leads to unbounded $T(t|t_0)$ and correspondingly unbounded concentration function for $k_{\mathrm d}=0$. This is impossible and unacceptable, since $T(t|t_0)$ and correspondingly the concentration function are bounded for any degradation reaction constant $k_{\mathrm d}$, given an impulsive point source.
For each integer value $n$ and by defining $\gamma=-D\lambda_n^2$, the principal solution for \eqref{sep5} is
\begin{equation}\label{Bessl}
R_n(r|r_{\rm tx})=E_nj_n(\lambda_{n}r)+F_n y_n(\lambda_{n}r),
\end{equation}
for any $\lambda_n$ value, where $j_n(\cdot)$ and $y_n(\cdot)$ are the $n$th order of the first and second types of spherical Bessel function, respectively. Since $y_n(\lambda_{n}r)$ is singular at $r=0$, we set $F_n=0$. Also, $R_n(r|r_{\rm tx})=E_nj_n(\lambda_{n}r)$ should satisfy the boundary condition \eqref{sep3B}. This implies $\lambda_n$ satisfies the following equation:
\begin{equation}\label{lambda}
D{\lambda _{n}}{j_n}'({\lambda _{n}}r_s) =-k_{\mathrm f}j_n({\lambda _{n}}r_s).
\end{equation}
Corollary \ref{cor2} from Theorem \ref{Th1} below implies that only the sequence of positive roots of \eqref{lambda} results in linearly independent solutions for \eqref{sep5} subject to the boundary condition \eqref{sep3B}.
We denote the $k$th positive root of the above equation and corresponding possible solution for \eqref{sep5} by $\lambda_{nk}$ and
\begin{equation}\label{R}
R_{nk}(r|r_{\rm tx})=E_{nk}j_n(\lambda_{nk}r),
\end{equation}
 respectively.
\begin{theorem}\label{Th1}
Let $\lambda _{vk}, k=0,1\ldots$, be the sequence of positive zeros of the third type boundary condition
\begin{equation}\label{lambda2}
r_s{\lambda _{vk}}{j_v}'({\lambda _{vk}}r_s) =-\zeta j_v({\lambda _{vk}}r_s), v+\zeta>0
\end{equation}
where $\zeta$ is a real constant. The system of spherical Bessel functions $ {r{j_v}({\lambda _{vk}}r)}, k=0,1 \ldots$, is orthogonal and complete  for $r\in [0,r_s]$, where orthogonality is defined as
\begin{equation}\label{Orth}
  \int\limits_0^{r_s} {{j_v}({\lambda _{vk}}r )} {j_v}({\lambda _{vk'}}r )r^2 dr  = \left\{ {\begin{array}{*{20}{c}}
{\begin{array}{*{20}{c}}
{{N_{vk}}}&{k = k'}
\end{array}}\\
{\begin{array}{*{20}{c}}
0&{k \ne k'}
\end{array}}
\end{array}} \right.
\end{equation}
where $N_{vk}=\frac{{{{r_s}^3}}}{2}(j_v^2({\lambda _{vk}}r_s) - {j_{v - 1}}({\lambda _{vk}}r_s){j_{v + 1}}({\lambda _{vk}}r_s))$.
\end{theorem}
\begin{proof}
The proof is presented in Appendix.
\end{proof}
\begin{corol}\label{cor1}
Let us set $v=n,\; n\in \mathbb{Z}_{+}$ and $\zeta=\frac{r_sk_{\mathrm f}}{D}$ in Theorem \ref{Th1}. Obviously $n+\frac{r_sk_{\mathrm f}}{D}>0$ for $n\in \mathbb{Z}_{+}$ and it is concluded that the system $r{{j_n}({\lambda _{nk}}r)}, k=0,1,\ldots$, where $\lambda _{nk}$ is the $k^{th}$ positive root of \eqref{lambda}, is orthogonal and complete for $r\in [0,r_s]$.
\end{corol}
\begin{corol}\label{cor2}
Based on Corollary \ref{cor1}, the functions $r{{j_n}({\lambda _{nk}}r)}, k=0,1,\ldots$ with only positive roots (${\lambda _{nk}}$) of \eqref{lambda} are linearly independent. It can then be shownthat ${{j_n}({\lambda _{nk}}r)}, k=0,1\ldots$, with positive roots (${\lambda _{nk}}$) of \eqref{lambda} constitute all linearly independent solutions for \eqref{sep5} subject to the boundary condition \eqref{sep3B}.
\end{corol}
%For positive $\gamma$ values, the principal solution for
%\eqref{sep5} is
%\begin{equation}\label{Bessl1}
%R_n(r|r_{\rm tx})=E_nj_n(\sqrt{-\gamma/D}r)+D_n y_n(\sqrt{-\gamma/D}r),
%\end{equation}
%However, we prove in \textbf{Appendix B} that this solution does not satisfy the boundary condition \eqref{sep3B} and is unacceptable.
%

To obtain $T(t|t_0)$, we consider the following ordinary differential equation from \eqref{sep4}
\begin{equation}\label{timeeq}
T'(t|{t_0})+(k_{\mathrm d}-\gamma)T(t|t_0)=0.
\end{equation}
 Given $\lambda_{nk}$ and considering the implicit condition of  $T(t\to \infty|t_0)=0$, the principle solution for \eqref{timeeq} is
 \begin{equation}\label{T}
T_{nk}(t|t_0)=I_{nk}e^{(-D\lambda_{nk}^2-k_{\mathrm d}) (t-t_0)}u(t-t_0),
\end{equation}
where $I_{nk}$ is an unknown constant.

Considering \eqref{Eqs} and the obtained separated solutions of \eqref{Phi}, \eqref{Theta}, \eqref{R}, and
\eqref{T}, the principal solution of the primal diffusion equation \eqref{fick} subject to the boundary condition \eqref{BD1} is given by
\begin{align}\label{Cf}
C(\bar r ,t|\bar r_{\rm tx},{t_0}) = \sum\limits_{n = 0}^\infty  {\sum\limits_{m = 0}^n {\sum\limits_{k = 1}^\infty  {{H_{mnk}}\cos (m(\varphi  - {\varphi _{\rm tx}}))} } } \\
 \times P_n^m(\cos \theta ){j_n}({\lambda _{nk}}r){e^{(-D\lambda_{nk}^2-k_{\mathrm d}) (t-t_0)}u(t-t_0)},\nonumber
\end{align}
where $H_{mnk}=G_mA_{nm}E_{nk}I_{nk}$ is an unknown constant and should be determined by applying the initial condition given in \eqref{sepin}. We note that for $m> n$, $P_n^m(\cos\theta)=0$ but we have only considered $m\leq n$ in the series.

To determine $H_{mnk}$, we expand the Dirac delta functions ${\delta (\varphi  - {\varphi _{\rm tx}})}$, ${\delta (\theta-\theta_{\rm tx})}$, and $\delta (r - {r_{\rm tx}})$ in initial condition \eqref{sepin} based on Fourier, Legendre, and Bessel series, respectively.
The function ${\delta (\varphi  - {\varphi _{\rm tx}})}$ in the interval $0<\varphi<2\pi$ can be represented by the Fourier series \cite[Eq. (5.46)]{Mandelis}
 \begin{equation}\label{dphi}
 {\delta (\varphi  - {\varphi _{\rm tx}})} =\sum\limits_{m=0}^\infty  {L_m \cos (m(\varphi  - \varphi _{\rm tx}))},
\end{equation}
where $L_0=\frac{1}{2\pi}$ and $L_m=\frac{1}{\pi}, m\geq 1$. The function ${\delta (\theta-\theta_{\rm tx})}, 0\leq \theta \leq \pi$ is expanded based on the orthogonal complete basis of Legendre functions $P_n^m(\cos {\theta})$, $n=0,1\ldots$, as follows \cite{GF}:
 \begin{equation}\label{dtheta}
 {{\delta (\theta  - {\theta _{\rm tx}})}} =\sin \theta  \sum\limits_{n = 0}^\infty   \frac{{2n + 1}}{2}\frac{{(n - m)!}}{{(n + m)!}} {P_n^m(\cos {\theta _{\rm tx}}) P_n^m(\cos \theta )}.
 \end{equation}
 %To expand the function $\delta (r - {r_{tx}})$ based on Bessel functions $r{{j_n}({\lambda _{nk}}r)}, k=0,1,\ldots$, where $\lambda _{nk}$ is the $k^{th}$ root of \eqref{lambda}, we have to prove the orthogonality and completeness of the system $r{{j_n}({\lambda _{nk}}r)}, k=0,1,\ldots$, where $\lambda _{nk}$ is the $k^{th}$ root of \eqref{lambda}. To this end, we first provide the following theorem.
%
%\begin{corol}
%Let us set $v=n$, $\lambda_{k}=\lambda_{nk}$, $\zeta=\frac{r_sk_{\mathrm f}}{D}$ in Theorem 1. Obviously $n+\frac{r_sk_{\mathrm f}}{D}>0$ and it is concluded that the system the system $r{{j_n}({\lambda _{nk}}r)}, k=0,1,\ldots$, where $\lambda _{nk}$ is the $k^{th}$ root of \eqref{lambda} is orthogonal and complete for $r\in [0,r_s]$.
%\end{corol}
Considering Corollary \ref{cor1}, ${\delta (r - {r_{\rm tx}})}, 0\leq r \leq r_s$, can be expanded based on the orthogonal and complete system $r{{j_n}({\lambda _{nk}}r)}, k=0,1\ldots$, as follows:
 \begin{equation}\label{dr}
 \delta (r - {r_{\rm tx}}) = \sum\limits_{k = 1}^\infty  {w_{nk}r^2} {{j_n}({\lambda _{nk}}{r})},
 \end{equation}
where
\begin{equation}\label{Nmn}
w_{nk}=\frac{\int_{0}^{r_s} \delta (r - {r_{\rm tx}})  {j_n}({\lambda _{nk}}{r})r^2dr}{\int_{0}^{r_s}   {j_n^2}({\lambda _{nk}}{r})r^2dr}= \frac{{j_n}({\lambda _{nk}}r_{\rm tx})}{\frac{{r_s^3}}{2}({j_n}({\lambda _{nk}}{r_s}) - {j_{n - 1}}({\lambda _{nk}}{r_s}){j_{n + 1}}({\lambda _{nk}}{r_s}))}\;\;.
\end{equation}

Substituting $C(\bar r ,t=t_0|\bar r_{\rm tx},{t_0})$ from \eqref{Cf}, Delta functions from \textbf{\eqref{dphi}-\eqref{dr}} in the initial condition \eqref{sepin}, and comparing left and right sides of the equation, we obtain

 \begin{equation}
 H_{mnk} = {L_m}\frac{{2n + 1}}{2}\frac{{(n - m)!}}{{(n + m)!}}\frac{P_n^m(\cos {\theta _{\rm tx}}){{j_n}({\lambda _{nk}}{r_{\rm tx}})}}{{{\frac{{r_s^2}}{3}({j_n}({\lambda _{nk}}{r_s}) - {j_{n - 1}}({\lambda _{nk}}{r_s}){j_{n + 1}}({\lambda _{nk}}{r_s}))} }}.
  \end{equation}
In the following, we remark on two properties of the obtained CGF.
\begin{remark}
Examining the CGF given in \eqref{Cf} reveals its \textit{reciprocity property}. In fact, the CGF does not change by interchanging the location of the observation point $(r,\theta,\varphi)$ and the point source transmitter $(r_{\rm tx},\theta_{\rm tx},\varphi_{\rm tx})$. This leads to the reciprocity of the corresponding DMC channel which may be exploited when analyzing and designing DMC networks.
\end{remark}

\begin{remark}
When the transmitter is located at the origin, the problem has elevation and azimuthal symmetry and the CGF is independent of $\phi$ and $\theta$ coordinates. In this case, the diffusion problem \eqref{Eq} simplifies to
\begin{equation}
  \frac{D}{{{r^2}}}\frac{\partial }{{\partial r}}({r^2}\frac{{\partial C(r,t|{t_0})}}{{\partial r}})
 - {k_{\mathrm d}}C(r,t|{t_0})+\frac{{\delta (r)\delta(t-t_0)}}{ r^2 } = \frac{{\partial C(r,t|{t_0})}}{{\partial t}},
\end{equation}
%By separation of variables of $C(r,t)=R(r)T(t)$, \eqref{Eqs} can be rewritten as
%\begin{equation}\label{Eqs2}
%  D(\frac{2}{r}\frac{{R'(r|{r_{tx}})}}{{R(r|{r_{tx}})}} + \frac{{R''(r|{r_{tx}})}}{{R(r|{r_{tx}})}}){{{r^2}}} = \frac{{T'(t|{t_0})}}{{T(t|{t_0})}} + k_{\mathrm d} \overset{c}{=}\gamma.
%\end{equation}
%\eqref{Eqs2} is the same as \eqref{sep4} when $n=0$. Thus , through
By the same procedure used above to derive the CGF, we obtain the CGF in this special case as
 \begin{equation}
C(r,t|{t_0}) = \sum\limits_{k = 1}^\infty  {{j_0}({\lambda _k}r){e^{ - {\lambda _k}^2D(t - {t_0})}}} ,
\end{equation}
where $\lambda_k$ is the $k$th root of the following equation:
\begin{equation}
  D(\lambda_k {j'_0}(\lambda_k r_s)=-k_{\mathrm f}j_0(\lambda_k r_s).
\end{equation}
Analogously, when the receiver is located at the origin and the transmitter is at an arbitrary location with radius $r_{\rm tx}$, the reciprocity property implies that the CGF is given by
  \begin{equation}
C(r=0,t|{t_0}) = \sum\limits_{k = 1}^\infty  {{j_0}({\lambda _k}r_{\rm tx}){e^{ - {\lambda _k}^2D(t - {t_0})}}}.
\end{equation}
\end{remark}

\section{Characterization of Received Signal}
In this section, the received signal at the receiver is characterized by employing the obtained CGF. Finally, the error probability of DMC with a simple on-off keying modulation over this channel is derived. The results in this section are adapted from our previous work \cite{Zoofaghari18} which is necessary to support the results that will come in Section V.

Based on our analysis in the previous section, and assuming an impulsive point source, the CGF $C(\bar{r},t|\bar{r}_{\rm tx},t_0)$ is given by \eqref{Cf}. %
We note that the differential equation system in \eqref{fick} with source input $S(\bar r,t,{\bar{r}_{\rm tx}},t_0)$ and output $C(\bar{r},t|\bar{r}_{\rm tx},t_0)$ is linear and time invariant. Therefore, given an arbitrary transmitter (not necessarily a point source or with instantaneous release) of $S(\bar r,t), \bar r \in \Omega$, the concentration at an arbitrary observation point $\bar{r}=(r,\theta,\varphi)$ is obtained as
\begin{equation}\label{out}
\iiint_\Omega C(\bar r,t|\bar r',t_0=0)*S(\bar r',t) r'^2 \sin \theta' dr' d\theta' d\varphi',
\end{equation}
where $*$ is the convolution operator and $C(\bar r,t|\bar r',t_0=0)$ is given by \eqref{Cf}. For a point source transmitter located at $\bar r_{\rm tx}$ with molecule release rate of $s(t)$, $S(\bar r',t)=s(t) \frac{{\delta (r'  - {r _{\rm tx}})\delta (\theta'-\theta_{\rm tx})\delta (\varphi'  - {\varphi _{\rm tx}})}}{r'^2 \sin\theta'}$, \eqref{out} simply reduces to $s(t)*C(\bar{r},t|\bar{r}_{\rm tx},t_0)$.

To derive the probability density, consider a point source transmitter located at $\bar{r}_{\rm tx}=(r_{\rm tx},\theta_{\rm tx},\varphi_{\rm tx})$ and a transparent receiver where the set of points inside the receiver is denoted by $\Omega_{\rm rx}$. Given the CGF \eqref{Cf}, the probability density function (PDF) of observation of a molecule, released from the point source transmitter at time $t_0$ inside a transparent receiver at time $t$ is obtained as
\begin{align}\label{eqlemm2}
p_{\rm obs}(t)=\iiint_{\Omega_{\rm rx}}{{C(\bar {r},t|{{\bar {r}}_{\rm tx}},{t_0})} r^2 \sin \theta dr d\theta d\varphi},
\end{align}
%where $\Omega_{\rm rx}$ denotes the geometric location of points inside the transparent receiver.
For a spherical receiver with small radius $R_{\rm rx}$ compared to the distance between receiver and transmitter, the concentration variation is negligible inside the receiver. Therefore, the probability density of the observation time given in \eqref{eqlemm2} is approximated by
\begin{align}\label{POBS}
\frac{4\pi}{3}R_{\rm rx}^3 {C(\bar r_{\rm rx},t|{\bar r_{\rm tx}},{t_0})},
\end{align}
where $\bar r_{\rm rx}$ is the center of the receiver.

Assuming the average modulated signal $s(t)$  for $t\in [0, T_0]$, the release rate of molecules can be modeled as a Poisson process \cite{Ion},
\begin{align}\label{TXmodel}
\textbf{s}(t)\sim \mathrm{Poisson} (s(t)).
\end{align}
Thus, the number of the molecules observed at the receiver at time $t\in[0,T_0]$, $\textbf{y}(t)$, originating from the molecules released in interval $[0,T_0]$, follows the Poisson process of \cite{Ion}
\begin{align}
\textbf{y}(t)\sim \mathrm{Poisson} \left(s(t)*p_{\rm obs}(t)\right).
\end{align}

Similarly, the residual ISI from the previous time slots can be derived.
%We characterized the received signal at the receiver due to transmitted signal in the current time slot, in the last subsection. Now, we explain how the residual ISI from the previous time slots can be incorporated in the receiver output.
Let $j$ denote the time slot number such that $j=0$ refers to the current time slot $[0,T_0]$ and $j>0$ denotes a previous time slot $[-jT_0,-(j-1)T_0]$. We assume that the average modulated signal in time slot $j$ corresponding to the input symbol for transmission in this time slot is denoted by $s_j(t+jT_0)$.
We also assume that the diffusion channel has memory of length $M$ time slots. Then the total ISI affecting the receiver output originating from $M$ previously transmitted symbols in the current time slot, $\textbf{I}(t)$, follows the Poisson process \cite{Ion}
\begin{align}\label{ISI}
\textbf{I}(t)~\sim \mathrm{Poisson} \left(\sum_{j=1}^{M}s_j(jT_0+t)*p_{\rm obs}(jT_0+t)\right).
\end{align}
%which follows a Poisson distribution since the $\textbf{I}_j$ are mutually independent Poisson RVs for $j\in \{1,\ldots,M\}$.
Therefore, given the current transmitted modulated signal, $\textbf{s}_0(t)$, the receiver observation at sampling time $t_s$ in the current time slot is $\textbf{y}_R=\textbf{y}(t_s)+\textbf{I}(t_s)$ which is a Poisson distributed RV with mean
\begin{align}
{y}_R(t)&={s_0(t_s)*p_{\rm obs}(t_s)}+\sum_{j=1}^{M}{s_j(jT_0+t_s)*p_{\rm obs}(jT_0+t_s)}=\sum_{j=0}^{M}{s_j(jT_0+t_s)*p_{\rm obs}(jT_0+t_s)}.
\end{align}

\subsection{Simple On-off Keying DMC System}
To evaluate the DMC system in BSE, a simple on-off keying modulation scheme is considered where 0 and 1 are represented by releasing 0 and $N$ molecules (on average) by the transmitter, respectively.
The transparent receiver counts the number of molecules inside the receiver volume at sampling time $t_s$ (which maximize $p_{\rm obs}(t)$) in each time slot.
The receiver uses the observed sample to decide about the transmitted bit.

Given the transmitted bits $B_i={b}_i, i\in\{0,1,\ldots,M\}$, the average modulated signal in time slot $i\in\{0,1,\ldots,M\}$ is $s_i(t+iT_0)=Nb_i\delta(t+iT_0)$. As shown in the last subsection, $\textbf{y}_R$ is a Poisson distributed RV, i.e.,
\begin{align}
&{\rm Pr}(\textbf{y}_R=y| b_0, {b}_1,\ldots,{b}_M)=\frac{e^{-\mathbb{E}(\textbf{y}_R| b_0, {b}_1,\ldots,{b}_M)}(\mathbb{E}(\textbf{y}_R| b_0, {b}_1,\ldots,{b}_M))^y}{y!},
\end{align}
in which
\begin{align}
&\mathbb{E}(\textbf{y}_R| b_0, {b}_1,\ldots,{b}_M)=\sum_{i=0}^{M}{{b}_iN\delta(iT_0+t)*p_{\rm obs}(iT+t)}=\sum_{i=0}^{M}{{b}_iNp_{\rm obs}(iT_0+t)},
\end{align}
where ${\rm Pr}(\cdot)$ and $\mathbb{E}(\cdot)$ denote probability function and expectation operator, respectively.

For error probability analysis, a genie-aided decision feedback (DF) detector \cite{Mosayebi2014} is assumed where a genie informs the detector of the \emph{previously} transmitted bits, i.e., $\hat{B}_i=B_i,\;i=1,\ldots,M$. Given the correct values of the previously transmitted bits, i.e., $B_i={b}_i, i\in\{1,\ldots,M\}$, are known at the decoder and ${\rm Pr}(B_0=1)={\rm Pr}(B_0=0)=\frac{1}{2}$, the Maxiumum-A-Posteriori (MAP) detector for bit $B_0$ given receiving $Y=y$ molecules in the current time slot becomes
\begin{align}\label{map}
\hat{B}_0=\mathrm{arg}\max_{b_0\in\{0,1\}}{{\rm Pr}(\textbf{y}_R=y| b_0, {b}_1,\ldots,{b}_M)},
\end{align}
where $\hat{B}_0$ denotes the estimated transmitted bit in the current time slot. Practically, the previously transmitted bits $B_i=b_i,\; i\in\{1,\ldots,M\}$, are not known. Therefore, previous decisions $\hat{B}_i=\hat{b}_i,\; i\in\{1,\ldots,M\}$, have to be used in \eqref{map} instead.

Simplifying \eqref{map} leads to a threshold decision rule based on the receiver output in the current time slot, $y$, \cite{Ion} i.e., $\hat{B}_0=0$, if $y\leq   {\rm Thr}$, and $\hat{B}_0=1$, if $y> {\rm Thr}$, where
\begin{align}\label{map_thr}
\mathrm{Thr}=\frac{{Np_{\rm obs}(t_s)}}{\ln \left(1+\frac{{Np_{\rm obs}(t_s)}}{\sum_{i=1}^{M}{Nb_ip_{\rm obs}(iT_0+t_s)}}\right)}.
\end{align}
The error probability of this detector is given by
\begin{align}\label{errorprob}
P_{\rm error}=\left(\frac{1}{2}\right)^{M+1}\sum_{b_0,\ldots ,b_M} {{\rm Pr}(E|b_0,b_1,\ldots,b_M)},
\end{align}
where $E$ is an error event, and we have
\begin{align}\label{cond_error}
&{\rm Pr}(E|b_0,b_1,\ldots,b_M)=\nonumber\\
&\sum_{y \underset{b_0=0}{\overset{b_0=1}{\lessgtr}} {\rm Thr}}{\frac{e^{-\mathbb{E}(\textbf{y}_R| b_0,b_1,\ldots,b_M)}(\mathbb{E}(\textbf{y}_R| b_0, b_1,\ldots,b_M))^y}{y!}}.
 \end{align}
 \color{black}

%Since the decoder does not know the correct values of the previously transmitted bits, for simplicity the ISI from the previously transmitted symbols given in \eqref{ISI} is approximated as a Poisson RV with mean $
%\sum_{j=1}^{M}{\frac{1}{2}Np_{\rm obs}(jT+t)}$. By formulating the maximum a-posteriori (MAP) detector, a threshold decision rule based on the receiver output in the current time slot, $y$, is obtained, i.e., $\hat{B}_0=0$, if $y\leq   {\rm Thr}$, and $\hat{B}_0=1$, if $y> {\rm Thr}$, where
%\begin{align}\label{map_thr}
%\mathrm{Thr}=\frac{{Np_{\rm obs}(t_s)}}{\ln \left(1+\frac{{Np_{\rm obs}(t_s)}}{\sum_{j=1}^{M}{\frac{1}{2}Np_{\rm obs}(jT+t_s)}}\right)}.
%\end{align}
%The bit error rate (BER) of this detector is given by
%\begin{align}\label{errorprob}
%P_{\rm error}=\frac{1}{2}({\rm Pr}(E|b_0=0)+{\rm Pr}(E|b_0=1)),
%\end{align}
%where $E$ is an error event, and we have \cite{Zoofaghari18}
%\begin{align}\label{cond_error}
%&{\rm Pr}(E|b_0)=\sum_{y \underset{b_0=0}{\overset{b_0=1}{\lessgtr}} {\rm Thr}}{\frac{e^{-\mathbb{E}(\textbf{y}_R| b_0)}(\mathbb{E}(\textbf{y}_R| b_0))^y}{y!}}.
% \end{align}

\section{Simulation and Numerical Results}
In this section, the effect of system parameters on the observation time PDF for diffusion in the considered BSE is investigated. Moreover, the performance of the point-to-point DMC system in this environment is evaluated.
To confirm the proposed analysis of observation time PDF (and correspondingly the CGF), a particle based simulator (PBS) is used.
In the PBS, time is divided into time steps of $\Delta t$ s. In each time step, the molecule locations are updated following random Brownian motion. The molecules move independently in the 3-dimensional space where the displacement of a molecule in $\Delta t$ s is modeled as a Gaussian RV with zero mean and variance $2D\Delta t$, in each dimension (Cartesian coordinates). Considering the degradation reaction given in \eqref{deg1}, a molecule may be removed from the environment during a time step $\Delta t$ s, with probability $k_{\mathrm d}\Delta t$ \cite{Elka}.
The boundary is fully covered by the receptor proteins characterized by \eqref{deg2}. Therefore, if a molecule hits the boundary, the molecule may bind with receptor $R$ and produce complex $AR$ with probability $k_{\mathrm f} \sqrt{\frac{\pi \Delta t}{D}}$ and may be reflected with probability of $1-k_{\mathrm f} \sqrt{\frac{\pi \Delta t}{D}}$ \cite{Elka}. Employing this probability for simulating the boundary condition results in quantitatively accurate PBS, when the simulation time steps or binding coefficients are very small (more precisely $k_{\mathrm f} \sqrt{\frac{\Delta t}{2D}}\ll 1/ \sqrt{2\pi}$)\cite{Andrew10}.
%Also, a produced AR complex may be unbounded and the information molecule returns back to the environment with probability of $1-e^{-k_b\Delta t}$ .
%For PBS, we have employed Poiseuille flow model in terms of $\bar v(\bar r)=2v_{\rm eff}(1-\frac{\rho^2}{\rho_c^2}) \hat a_z$ $ms^{-1}$. Accordingly, the movement of the molecule at a point with radial $\rho$ due to the flow equals to $2v_{\rm eff}(1-\frac{\rho^2}{\rho_c^2}) \Delta t$. Although, we assumed a constant flow velocity of $v \hat a_z$ to have a tractable analysis, our results indicates that the proposed analysis confirms PBS, when average velocity of $v=\int_{0}^{\rho_c} {2v_{\rm eff}(1-\frac{\rho^2}{\rho_c^2})dr}=\frac{4}{3}v_{\rm eff}$ is adopted in the analysis.
\begin{table}
	\begin{center}
		 		 \caption{ِDMC system parameters used for analytical and simulation results} \label{table1}
		\begin{tabular}{|c|c|c|}
			\hline
			\bfseries{Parameter} & \bfseries{Variable} & \bfseries{Value}   \\
			\hline \hline
			Diffusion coefficient & $D$ & $ 10^{-9}\;\si{m^2.s^{-1}}$\\
			\hline
			Sphere radius & $r_s$ & $5, 6, 7, 10, \infty$ $\si{\mu m}$\\
			\hline
			Point source transmitter location & $(r_{\rm tx},\theta_{\rm tx},\varphi_{\rm tx})$ &\textbf{ $(3\si{\mu m},\pi/2,0),(0.25\si{\mu m},\pi/2,0)$}\\
			\hline
			Degradation reaction constant inside   & $k_{\mathrm d}$ & $0, 20$ $s^{-1}$\\the sphere & &  \\
			\hline
			Ligand-receptor reaction constant over & $k_{\mathrm f}$ & $0, 100, \infty$ $\si{\mu m.s^{-1}}$\\ the surface & & \\
			
			\hline
			Receiver radius & $R_{\rm rx}$ & $1 \si{\mu m}$ \\
			\hline
			Number of transmitted molecules for bit `1'& $N$ & $5 \times 10^4$ \\
			\hline
			Time step in PBS& $\Delta t$ & $10^{-5} \si{s}$ \\
			\hline
		\end{tabular}
	\end{center}
\end{table}
%----------------------------------------------------------------------------------------
The point source transmitter is located at $(r_{\rm tx},\theta_{\rm tx},\varphi_{\rm tx})=(3\si{\mu m},\pi/2,0)$ and the diffusion coefficient is $D=10^{-9} \si{m^2.s^{-1}}$. The system parameters used for all of the analytical and simulation results are presented in Table \ref{table1}.

Fig. \ref{Fig1} compares the observation time PDF obtained from our analysis given in \eqref{POBS} and PBS, when the receiver center is located at $r_{\rm rx}=4\,\si{\mu m}$ with different elevation and azimuth coordinates $\theta_{\rm rx}=\{\pi/4,\pi/2\}$ and $\varphi_{\rm rx}= \{0,\pi/2,3\pi/4\}$
 when $r_s=5 \si{\mu m}$, $k_{\mathrm f}=100$ $\si{\mu m.s^{-1}}$, and $k_{\mathrm d}=20 s^{-1}$. It is observed that the PBS confirms the proposed analysis, capturing the PDF variations in azimuth and elevation coordinates in addition to the radial coordinate. Also, Fig. \ref{Fig1} depicts the observation time PDF obtained from analysis for a receiver located at $r_{\rm rx}=4\si{\mu m}$ and different elevation and azimuth coordinates, when the transmitter is located closer to the origin, i.e., $(r_{\rm tx},\theta_{\rm tx},\varphi_{\rm tx})=(0.25\si{\mu m},\pi/2,0)$. Comparing with the PDFs for the two transmitter distances(i.e.,$r_{\rm tx}=0.25 \si{\mu m}$ and $r_{\rm tx}=3 \si{\mu m}$), we deduce that PDF variation in elevation and azimuth coordinates decreases when the transmitter becomes close to the origin. This occurs because the elevation and azimuth symmetry increases when the transmitter is closer to the origin. Obviously, the transmitter that is located exactly at the origin leads to perfect symmetry with respect to the elevation and azimuth coordinates.

In Fig. \ref{Fig2}, the observation time PDF obtained from the analysis in \eqref{POBS} and the PBS for different spherical environment radius values $r_s=\{5,6,7,10\} \si{\mu m}$ and also unbounded environment are compared, when the receiver is located at $(4\si{\mu m}, \pi/4,3\pi/4)$, $k_{\mathrm f}=100$ $\si{\mu m.s^{-1}}$, and $k_{\mathrm d}=20 s^{-1}$. We observe that the PBS confirms the analytical results. Fig. \ref{Fig2} also shows that the observation PDF is significantly amplified for smaller sphere radius values. Moreover, we observe that the unbounded approximation may still be valid and useful for a sufficiently large environment radius.

For the different scenarios used in Fig. \ref{Fig2}, the performance of simple on-off keying DMC system in terms of bit error rate is shown in Fig. \ref{Fig7} when $N=5\times 10^4$ and the center of the transparent spherical receiver with radius $R_{\rm rx}=0.5 \si{\mu m}$ is located at $(4\si{\mu m},\pi/4,3\pi/4)$. The receiver observes the number of molecules inside its volume at the sampling time at which the observation probability is maximized.
The BER for different scenarios obtained from \eqref{errorprob} has been depicted versus time slot duration, $T_0$, when the channel memory is adopted $0.2$ s (correspondingly $M=0.2/T_0$ bits). The analytical BERs are verified by a Monte Carlo simulation with $10^7$ bits in which the received signal in each time slot at the receiver is generated based on the presented model and not by following particle movements.\footnote{A Monte Carlo simulation has been employed for verifying BER results, since using the PBS takes very long time for large number of bits (here $10^7$ bits).} As expected, the BER is a decreasing function of time slot duration, because for a shorter time slot duration (higher transmission rate), a higher memory and more ISI is encountered. It is observed that the BER increases and approaches the BER of unbounded environment for radius values higher than a threshold which is around 7 \si{\mu m} for the adopted parameters. The threshold depends on the positions of the transmitter and receiver, i.e., their distance from each other and from boundaries. In particular, the unbounded assumption can be adopted when the distance between transmitter and receiver is sufficiently smaller than their distances from the boundaries.

Fig. \ref{Fig4} depicts the observation time PDF in the presence of the degradation reaction with $k_{\mathrm d}=0$ and $20$, with different boundary conditions including absorbing boundary $(k_{\mathrm f}\to \infty$), reflective boundary ($k_{\mathrm f}=0$), partially absorbing ($k_{\mathrm f}=10^{-4}$) boundary, and unbounded environment ($r_s\to \infty$). The observation time PDF obtained from the PBS and analysis \eqref{POBS} has been depicted for a receiver located at $(4\si{\mu m},\pi/4,3\pi/4)$ when $r_s=5 \si{\mu m}$ and $\bar r_{\rm tx}=(3\si{\mu m},\pi/2,0)$. The PBS confirms the proposed analytical results, in all scenarios.

It is observed that both the degradation and the (partially) absorbing boundary attenuate the observation probability (correspondingly the gain of the diffusion channel) from one side and shorten the tail of the observation probability curve (correspondingly the memory of the diffusion channel) from the other side. As a result, a trade-off between the gain and memory of the diffusion channel exists in the presence of degradation and partially absorbing boundary. For instance, PDF for unbounded environment has higher amplitudes (and larger memory) compared to the absorbing boundary, since the molecules hitting the absorbing boundary are removed and do not return to the environment. On the other hand, the PDF for the unbounded environment has lower amplitude (and smaller memory) compared to the reflective boundary, since the diffusion of molecules is confined within the boundary when the boundary is reflective leading to the higher concentration and memory inside the sphere.

For the different scenarios used in Fig. \ref{Fig4}, the performance of the simple on-off keying DMC system in terms of bit error rate is shown in Fig. \ref{Fig5} when $N=5\times 10^4$ and the center of the transparent spherical receiver with radius $R_{\rm rx}=0.5 \si{\mu m}$ is located at $(4\si{\mu m},\pi/4,3\pi/4)$. The receiver observes the number of molecules inside its volume at the sampling time at which the observation probability is maximized.
The BER for different scenarios obtained from \eqref{errorprob} has been depicted versus time slot duration, $T_0$, and verified by Monte Carlo simulation.

As expected, the BER is a decreasing function of time slot duration, because for a shorter time slot duration (higher transmission rate), a higher memory and more ISI is encountered.
It is also observed that the BER for the partially absorbing boundary compared to the fully absorbing and reflective boundaries is lower. Comparing with the reflective boundary, the partially absorbing boundary has lower channel gain, but encounters less ISI as observed in Fig. \ref{Fig4}. In this comparison, the effect of ISI is dominant leading to improved BERs for the partially absorbing scenario. Compared to the fully absorbing boundary, the partially absorbing boundary has higher channel gain, while it encounters higher ISI as observed in Fig. \ref{Fig4}. In this case, the effect of channel gain is dominant and the result is a lower BER. This reveals the trade-off between the gain and memory of the diffusion channel resulting from the absorbing boundary, as discussed above.% In fact, the effect of channel memory is dominant for smaller $T$ values; thus absorbing boundary results smaller BER than reflective boundary. Also, the effect of channel gain is dominant for higher $T$ values; thus the reflective boundary case outperforms BER.

%%%%%%%%%%%%%%%%%%%%%%%%%%%%%%%%%%%%%%%%%%%%%%%%%%%%%%%%%%%%%%%%%%%%%%%%%%%%%%%%%%%%%%%%%
\section{Conclusion}
A BSE was considered for a DMC system in which the molecules are exposed to a degradation reaction inside and irreversible receptor proteins over the inner boundary of the environment sphere. The concentration Green's function of diffusion in this environment was analytically derived, which takes into account asymmetry in all radial, elevation, and azimuth coordinates. Correspondingly, the received signal at the receiver was characterized. The presented model and analysis can be used to predict the drug concentration profile in biological sphere-like entities for drug delivery applications. Based on our analysis, it was revealed that the information channel is reciprocal in the described environment. Furthermore, the provided analysis enables us to examine and validate the conventional unbounded environment assumption. To examine the communications performance of the DMC system in this biological sphere, a simple on-off keying modulation scheme was adopted. We observed how the degradation reaction and partially absorbing boundary may result in a trade-off between the channel gain and channel memory. Considering DMC in the biological sphere with reversible receptor proteins over the boundary is left for future work.

\appendix
\section{Proof of Theorem 1}\label{App1}
The spherical Bessel function of order $v$, $j_v(\lambda _{vk}r)$, is related to the cylindrical Bessel function of order $v+0.5$, ${J_{v + 0.5}}({\lambda _{vk}}r)$, as follows \cite[Eq.~(7.46)]{Andrews}:
\begin{equation}\label{SphCyl}
j_v(\lambda _{vk}r)=\sqrt {\frac{\pi}{{2{\lambda _{vk}}r}}} {J_{v+ 0.5}}({\lambda _{vk}}r).
\end{equation}
Substituting \eqref{SphCyl} into the boundary condition \eqref{lambda2} and with some simple manipulation, we obtain
\begin{equation}\label{BR2}
  r_s\lambda _{vk}J{'_{v + 0.5}}(\lambda _{vk}{r_s})=( - \zeta +0.5){J_{v + 0.5}}(\lambda _{vk}{r_s})
\end{equation}

Therefore, the system $r j_v(\lambda _{vk}r)$ with $\lambda_{vk}, k=0,1,\ldots$ as a sequence of roots of \eqref{lambda2}, is equivalent to the system $\sqrt{r}{J_{v + 0.5}}({\lambda _{vk}}r)$ with $\lambda_{vk}, k=0,1,\ldots$ as a sequence of roots of \eqref{BR2} that constitutes an orthogonal and complete system in $r\in[0,r_s]$ when $v+\zeta>0$ \cite[Ch.2]{Higgins}. This completes the proof.

\newpage

\begin{figure}
\center
\includegraphics[width=15 cm,height=9 cm]{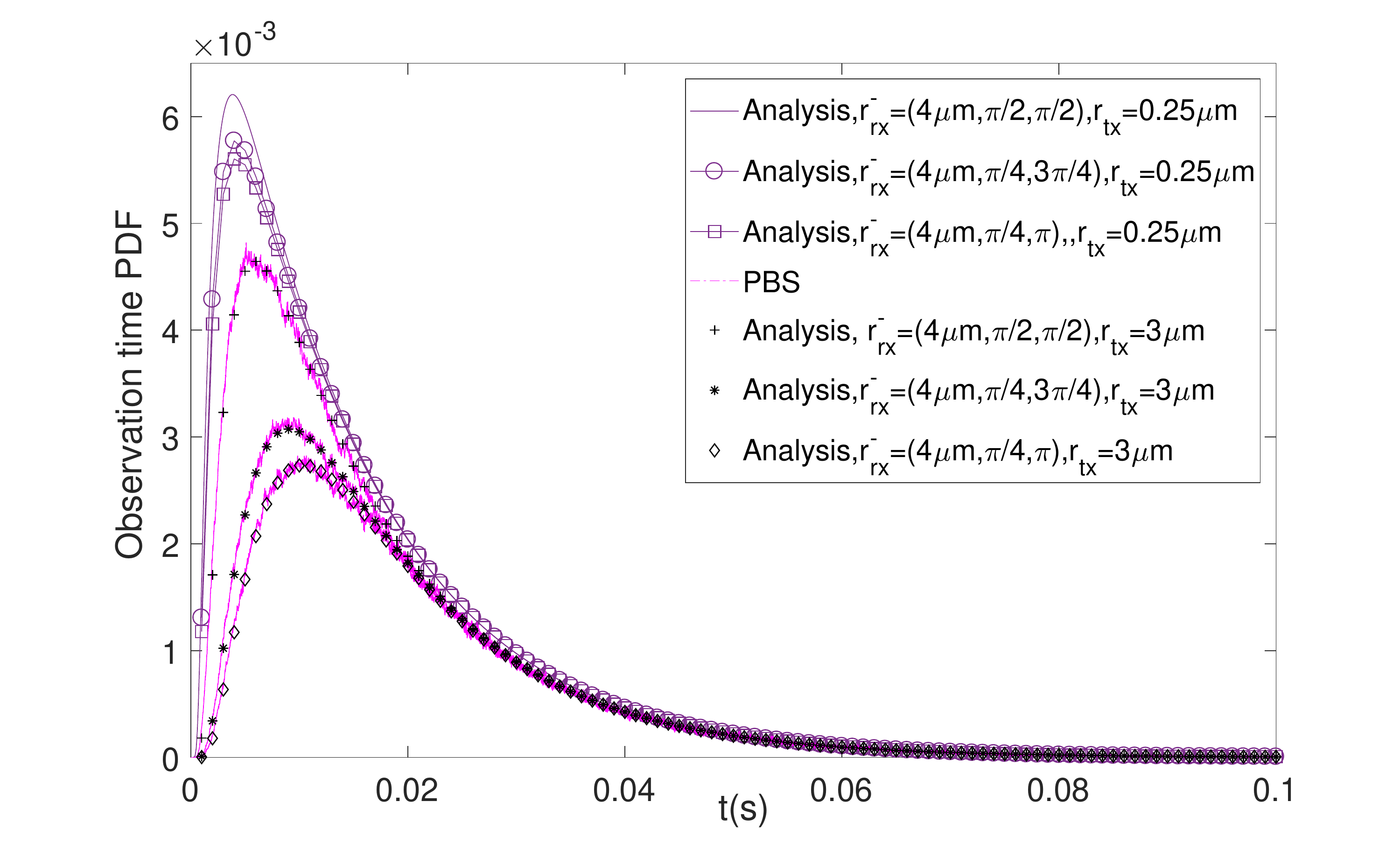}
	\setlength{\abovecaptionskip}{-0.5 cm}
 \caption{Observation time PDF obtained from analysis and PBS for different locations of observation point.}
\label{Fig1}
\end{figure}
\begin{figure}
\center
\includegraphics[width=15 cm,height=9 cm]{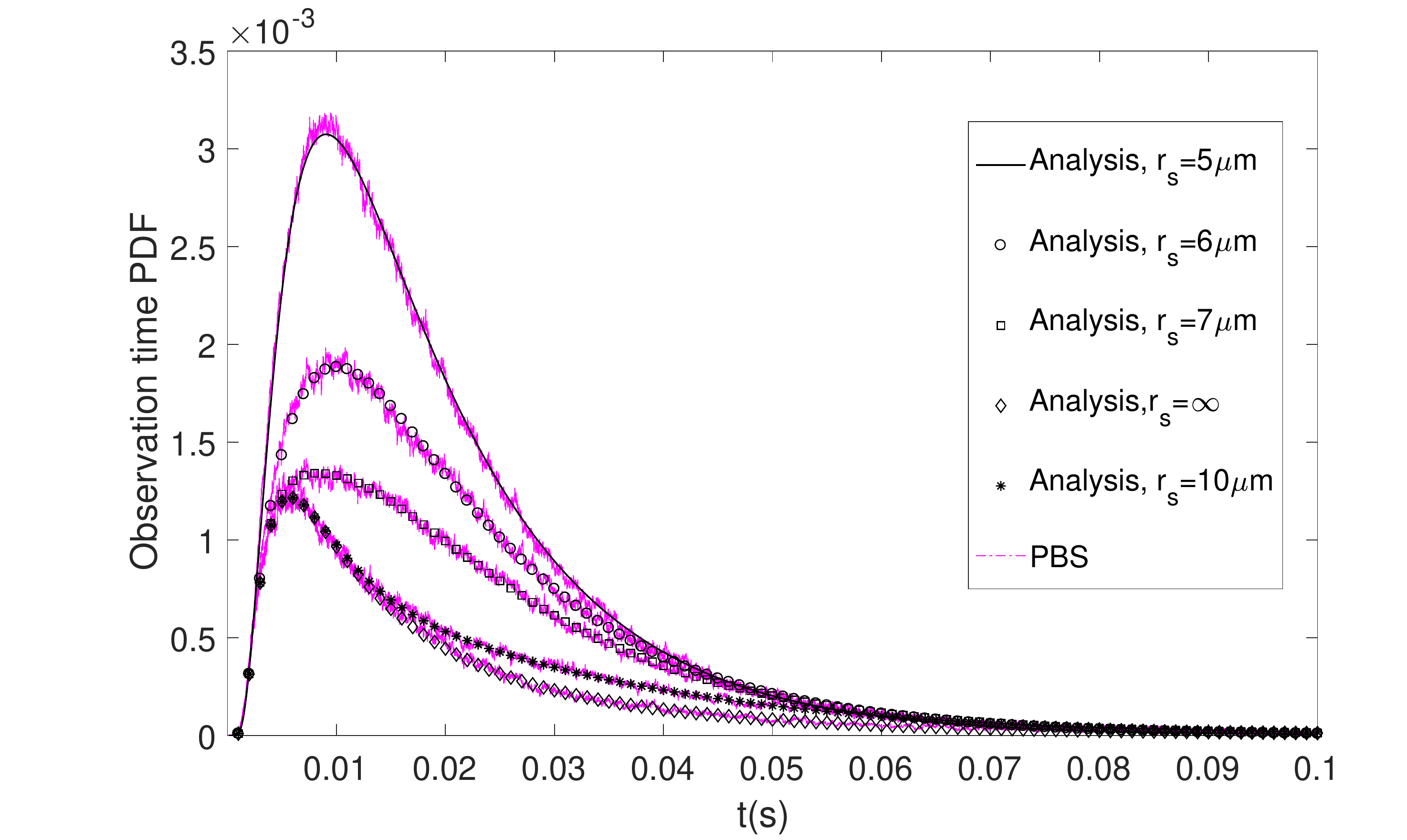}
	\setlength{\abovecaptionskip}{-0.5 cm}
 \caption{Observation time PDF obtained from analysis and PBS for different sphere radius values, $r_s=\{5,6,7,10,\infty\}$ $\si{\mu m}$.}
\label{Fig2}
\end{figure}

\begin{figure}
\center
\includegraphics[width=15 cm,height=9 cm]{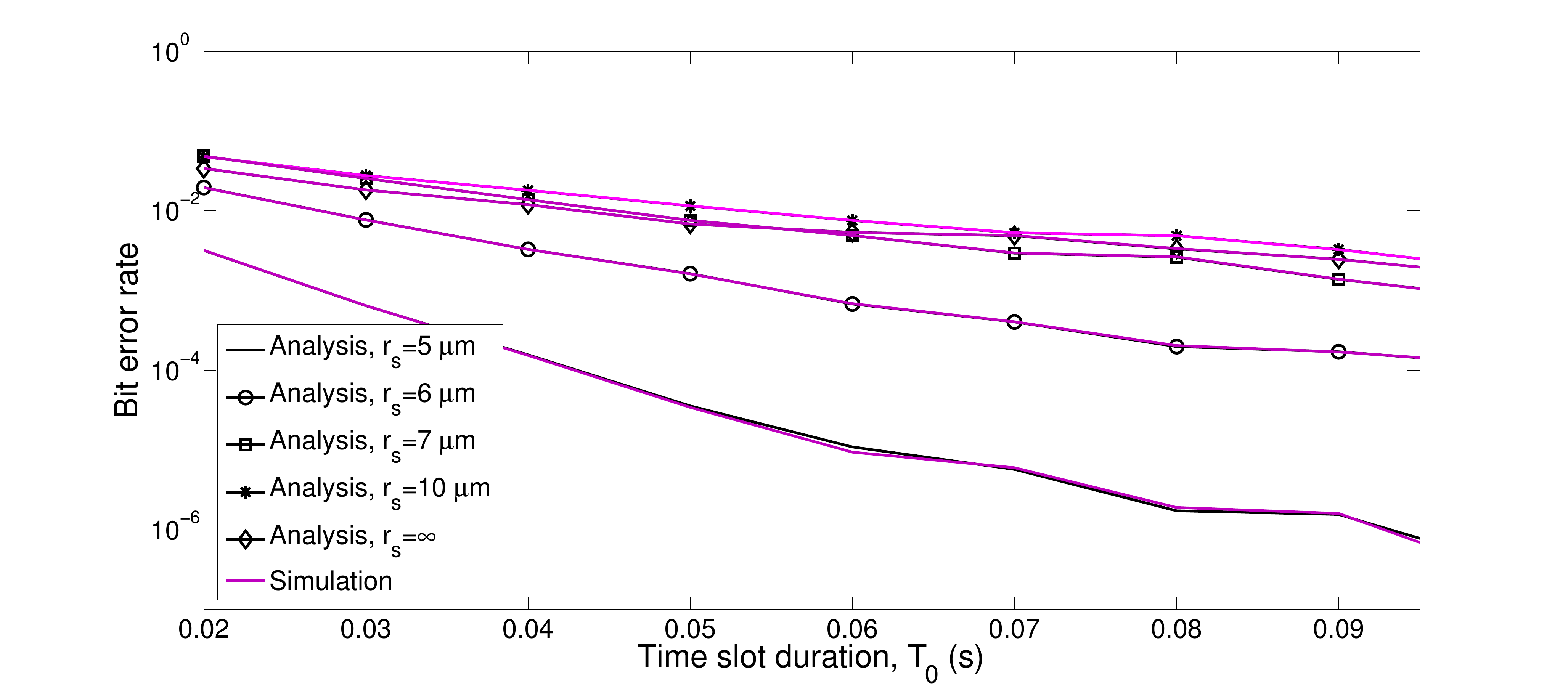}
	\setlength{\abovecaptionskip}{-0.5 cm}
 \caption{BER of the DMC system corresponding with scenarios in Fig. \ref{Fig2}. }
\label{Fig7}
\end{figure}

\begin{figure}
\center
\includegraphics[width=15 cm,height=9 cm]{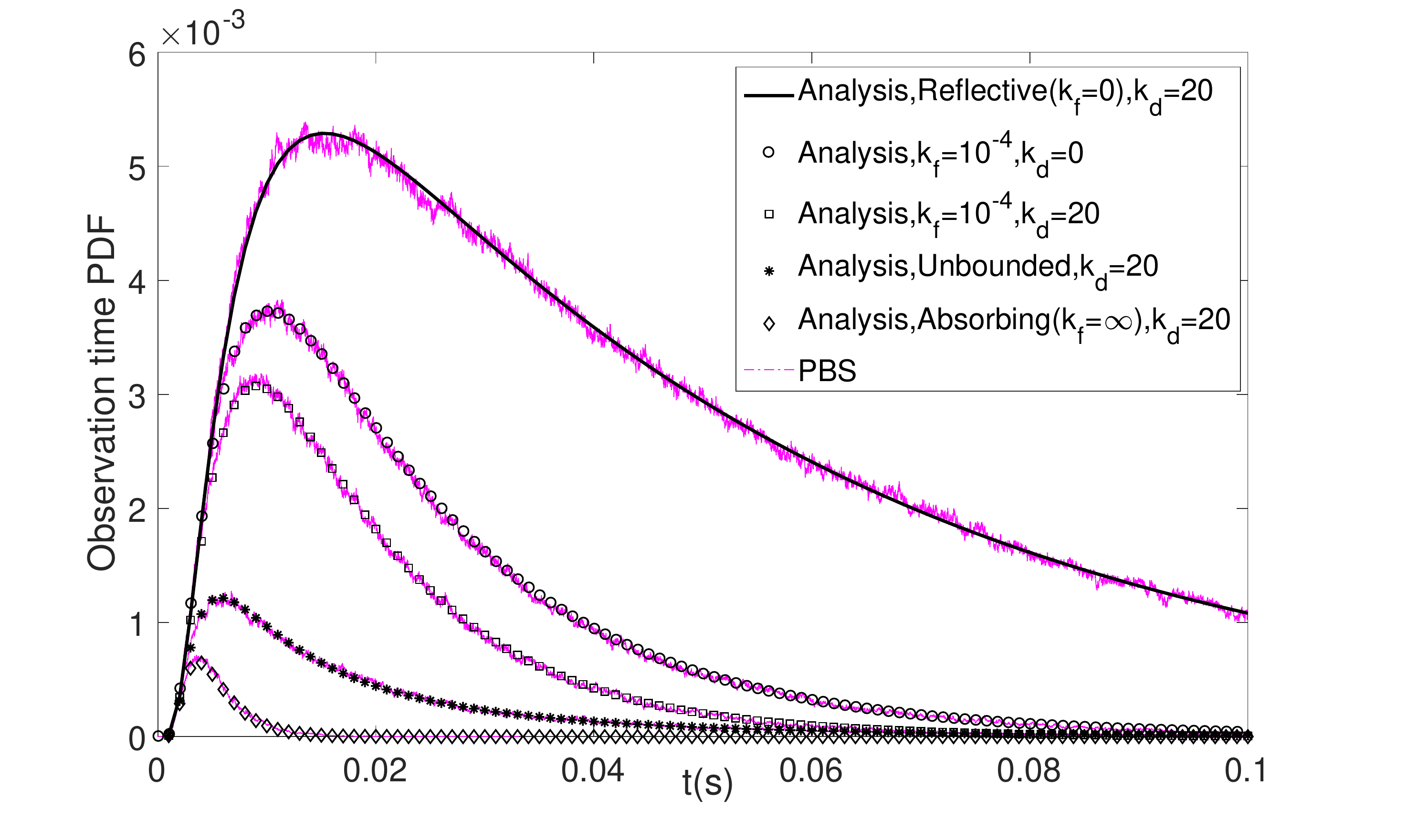}
	\setlength{\abovecaptionskip}{-0.5 cm}
 \caption{Observation time PDF obtained from analysis and PBS for diffusion in unbounded and spherical environment for different $k_{\mathrm d}$ and $k_{\mathrm f}$ values.}
\label{Fig4}
\end{figure}
\begin{figure}
\center
\includegraphics[width=15 cm,height=9 cm]{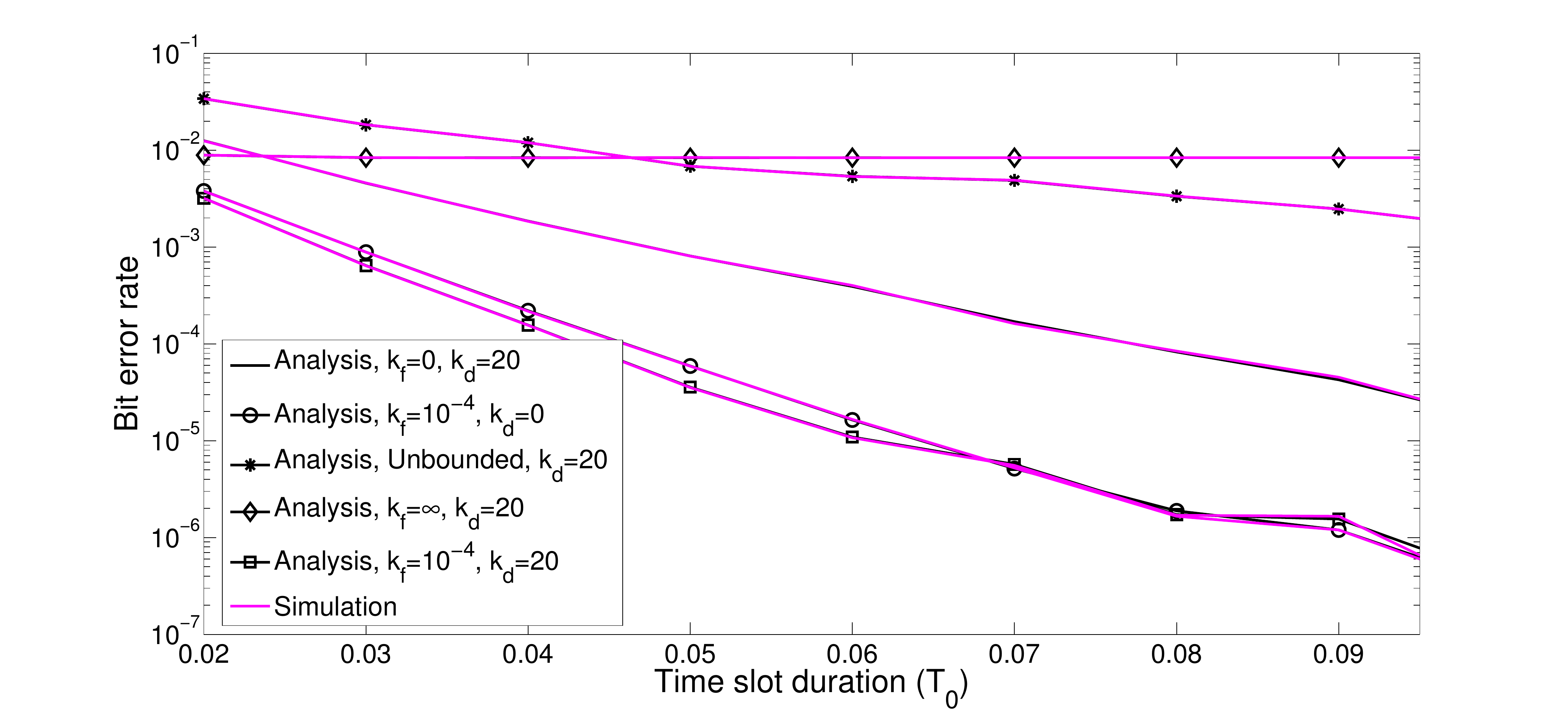}
	\setlength{\abovecaptionskip}{-0.5 cm}
 \caption{BER of the DMC system corresponding with scenarios in Fig. \ref{Fig4}. }
\label{Fig5}
\end{figure}

\end{document}